\documentclass[11pt]{article}

\usepackage[margin=1in]{geometry}

\usepackage{hyperref}
\pdfoutput=1 

\usepackage{amsthm}
\usepackage[cmex10]{amsmath}
\usepackage{amssymb}
\usepackage[caption=false,font=scriptsize]{subfig}

\newtheorem{theorem}{Theorem}
\newtheorem{lemma}[theorem]{Lemma}
\newtheorem*{definition}{Definition}
\usepackage[pdftex]{graphicx}
\usepackage{array}

\usepackage{balance}
\usepackage{mathptmx}

\usepackage{pbox}

\usepackage{algorithm}
\usepackage[noend]{algorithmic}
  
\usepackage{xspace}
\newcommand{\eps}{\ensuremath{\varepsilon}}
\newcommand{\ours}{{\sc MRG}\xspace}
\newcommand{\ene}{{\sc EIM}\xspace}
\newcommand{\gon}{{\sc Gon}\xspace}
\newcommand{\unif}{{\sc Unif}\xspace}
\newcommand{\gauss}{{\sc Gau}\xspace}
\newcommand{\unbalanced}{{\sc UnB}\xspace}

\newcommand{\poker}{{\sc Poker Hand}\xspace}
\newcommand{\kdd}{{\sc KDD Cup}\xspace}

\newcommand{\OPT}{\ensuremath{\mathit{OPT}}}

\usepackage{siunitx}

\begin{document}

\title{\Large Efficient Parallel Algorithms for $k$-Center Clustering}

\author{Jessica McClintock\\
{Department of Computing and Information Systems}\\
{The University of Melbourne}\\
{jmcclintock@unimelb.edu.au}
\and 
Anthony Wirth\\
{Department of Computing and Information Systems}\\
{The University of Melbourne}\\
awirth@unimelb.edu.au
}
\date{}
\maketitle

\begin{abstract} 
The $k$-center problem is one of several classic NP-hard clustering questions.
For contemporary massive data sets, RAM-based algorithms become impractical.
And although there exist good sequential algorithms for $k$-center, they are not easily parallelizable.

In this paper, we design and implement parallel approximation algorithms for
this problem.
We observe that Gonzalez's greedy algorithm can be efficiently parallelized in
several MapReduce rounds; in practice, we find that two rounds are sufficient,
leading to a $4$-approximation.
We contrast this with an existing parallel algorithm for $k$-center that runs
in a constant number of rounds, and offers a $10$-approximation.
In depth runtime analysis reveals that this scheme is often slow, and that its
sampling procedure only runs if~$k$ is sufficiently small, relative to the
input size.
To trade off runtime for approximation guarantee, we parameterize this
sampling algorithm, and find in our experiments that the algorithm is not only
faster, but sometimes more effective.
Yet the parallel version of Gonzalez is about~$100$ times faster than both its
sequential version and the parallel sampling algorithm, barely compromising
solution quality.
\end{abstract}

\medskip
\paragraph*{Keywords}
\noindent
Clustering, $k$-center, approximation algorithms, parallel algorithms.

\section{Introduction}

Clustering is a fundamental task in interpreting data sets in contexts such as social networking, event recognition and bioinformatics.
For many applications, the data sets can be prohibitively large, and there may be
insufficient RAM to perform the necessary calculations efficiently,
even when seeking approximate solutions. 
There are parallel-computing schemes such as MapReduce~\cite{dean2008mapreduce}
that offer the ability to overcome the memory obstacle.

The $k$-center problem is a famous clustering problem, is NP-hard, and has
well known (sequential) polynomial-time algorithms that offer essentially the best
approximation possible.
We describe a multi-round parallel algorithm for $k$-center, analyze in detail several
parallel algorithms and compare them with
one of these sequential methods.
Inspired by theoretical guarantees and evaluation, we run comprehensive
experiments, including trading off approximation for running time.

\subsection{Clustering algorithms}

Generally, clustering problems involve optimizing some function that indicates how well the clusters
portray underlying structures in the data.
In a \emph{metric} clustering problem, the
weights, representing the similarity between objects,
observe the triangle inequality.
The best-known example is of course the Euclidean metric.
In the context of clustering, points in a metric space can be modelled as vertices in a complete graph.
Each vertex stands for a data point, and each edge is weighted to
indicate the distance (or dissimilarity) between the two adjacent points.
The $k$-center problem is one of the fundamental NP-hard clustering problems on a metric input.

\begin{definition}[$k$-\textsc{Center}]
Find a set of at most~$k$ centers ---
here we assume they are a subset of the vertices ---
such that the maximum distance from a vertex to its assigned center
is minimized.
The key task is to choose the optimum set of $k$~centers, as each of the
remaining vertices would be assigned to its nearest center.
For a set of points~$V$, solution set~$S$ containing at most~$k$ vertices,
and a distance function~$d$, the objective of this problem can be considered as minimizing the objective $\min_{v\in V}\max_{s\in S}d(v, s)$.
\end{definition}

This objective has many applications, from vehicle routing to document clustering, in which it relates to concepts such as the furthest traveling time, or the least ``similar'' document. 
It can alternatively be considered to be  minimizing the (maximum)
covering \emph{radius} of the clusters.
Related classic NP-hard clustering questions include
$k$-median, $k$-means, and facility location problems.

 Via a reduction from the \textsc{Dominating set} problem, 
 Hsu and Nemhauser proved that, for all~$\eps > 0$, it is NP-hard to guarantee approximations within a factor $2-\eps$ of optimum for $k$-center~\cite{hsu1979easy}. 
Also exploiting the connection with the dominating set problem, Hochbaum and Shmoys gave a $2$-approximation algorithm for the $k$-center problem~\cite{hochbaum1985best}.
Gonzalez introduced an greedy $2$-approximation algorithm for the $k$-center problem~\cite{gonzalez1985clustering}.
Each of these $k$-center algorithms is inherently sequential,
none admitting a simple parallel implementation. 

\subsection{Parallel algorithms}

While approximation algorithms provide guaranteed performance with 
polynomial-time complexity, often data sets are large enough that running
these algorithms efficiently requires prohibitively large amounts of RAM. 
For such instances, we can instead design algorithms that split the data across multiple machines, and process each part in parallel before aggregating the results. 
One important paradigm for parallel computing is
MapReduce~\cite{dean2008mapreduce}. 
There are several fast approximation algorithms for
famous clustering problems, such as $k$-center and $k$-median, in
MapReduce~\cite{blelloch2010parallel,chierichetti2010max,ene2011fast}.
Karloff et al.~\cite{karloff2010model} introduced a theoretical
model of computation for the MapReduce paradigm that is often applied
to the analysis of MapReduce algorithms~\cite{bahmani2012scalable,lattanzi2011filtering}.
They offered a comprehensive method for theoretically structuring algorithms for MapReduce,
and defined a family of classes for MapReduce algorithms.

A MapReduce algorithm consists of a series of interleaving rounds of sequential \emph{mappers} and parallel \emph{reducers}.
A map round assigns each data point independently to some reducer(s);
the reducers run in parallel, each performing some procedure
on the subset of points it has been assigned.
A program in MapReduce may consist of several iterations
of mappers and reducers, each involving potentially different map and reduce functions.

Both the $k$-means~\cite{bahmani2012scalable} and $k$-median~\cite{ene2011fast} problems have been adapted to the MapReduce framework.

\subsection{Our contribution}
We provide a very careful and detailed examination both of the best-known
MapReduce approximation algorithm for $k$-center~\cite{ene2011fast}, based on
sampling,
and a parallel implementation of Gonzalez's algorithm (that typically gives a
$4$-approximation).
The $2$-round special case of the latter approach was recently considered by Malkomes et
al.~\cite{malkomes2015fast}, although their analysis and experiments differ
considerably from ours.
We describe in depth the performance and computational requirements of these
approaches, and detail how this procedure can be adapted to allow for cases where RAM is insufficient even for the $2$-round parallel solution.
Based in part on a careful calculation of its running time,
we generalize the sampling MapReduce scheme of Ene et al.~\cite{ene2011fast},
to trade off approximation guarantee for speed.

Our experiments show that the parallelized Gonzalez approach is often~100 times faster than the alternatives, while being almost as effective.
These are the first experimental results for the $k$-center algorithm of Ene et al.~\cite{ene2011fast}.
Our results conform with the findings of Malkomes et
al.~\cite{malkomes2015fast}, regarding the performance of their greedy approach.

\section{Related work}

\subsection{Approximations}
The $k$-center problem was first adapted to the MapReduce scheme
by Ene et al.~\cite{ene2011fast}.
Their algorithm selects a uniformly random sample of the points,~$S$, and
adds points to~$S$ until most vertices of the graph are
within a bounded distance of the sample.
Finally, it adds remaining unrepresented vertices to~$S$. 
A sequential $k$-center algorithm, with approximation factor~$\alpha'$,
is then run on~$S$.
With high probability,
the~$k$ resulting centers constitute
a $5\alpha'$-approximation for the $k$-center instance.
When implemented using one of the $2$-approximation algorithms described
above, with high probability, this results in a $10$-approximation overall. 
Ene et al.~\cite{ene2011fast} apply a similar scheme to the $k$-median
problem, with an $11\alpha^*$-approximation (where~$\alpha^*$ is a the factor
of the standard approximation for $k$-median).

Recently, there has been increased interest in adapting $k$-center
to Map\-Reduce. 
Ceccarello et al.~\cite{DBLP:CeccarelloPPU14} gave
a MapReduce diameter-approximation procedure with low parallel depth.
From this, they derive a
$k$-center solution for graphs with unit-weight edges: 
for $k\in \Omega(\log^2 n)$,
with high probability, this is a $O(\log^3 n)$-approximation. 
Im and Moseley~\cite{im2015} have described a
randomized $3$-round $2$-approximation algorithm that requires
prior knowledge of the value of the optimal solution.
Although they have announced that this leads
to a $4$-round $2$-approximation without the requirement,
the details have yet to be outlined.
Very recently, Malkomes et al.~\cite{malkomes2015fast} gave a $2$-round approach similar to ours.

\subsection{Experiments}
Ene et al.~\cite{ene2011fast} reported that their $k$-center MapReduce scheme
performs poorly due to the sensitivity of $k$-center to sampling.
Unfortunately, there are no results nor implementation details to confirm this.
In combination with another simpler algorithm, we investigate the empirical performance their $k$-center scheme in greater detail.

Conversely, their $k$-median implementation performs significantly better than the worst-case guarantee.
Solutions are comparable to sequential algorithms with much better bounds. 
Ene et al.'s results were based on the $3$-approximation algorithm
from Arya et al.~\cite{arya2004local}.
There have been recent advances in $k$-median approximation algorithms, by Byrka et al.~\cite{DBLP:ByrkaPRST15}, as well as by Li and Svensson~\cite{li2013approximating}, and including these might improve the approximation bound.

\section{Parallel \emph{k}-center}

We describe and analyze an approximation algorithm for the $k$-center problem that, for most practical cases, achieves a $4$-approximation in only two MapReduce rounds.
The intuition is that a sequential $k$-center algorithm finds in the first
round a sample from  each of the reducers
such that the distance to all of the unsampled points is bounded.
Running a standard factor-$2$ algorithm on the sample
reveals a factor-$4$ solution to the whole instance.
Additional rounds can be performed in cases where even the sample is too large
for a single machine:
this would usually occur for very large values of~$k$. 
Experiments show that this approach is often as good as that of the baseline
sequential algorithm.
The $2$-round case of our algorithm is similar to the approach of Malkomes et
al.~\cite{malkomes2015fast}.
Along with generalizing to larger instances, we analyze the run time of these
algorithms in more detail,
and provide an alternative, shorter proof of the two-round factor-four
approximation.

\subsection{Description}

In this paper, the standard~$k$-center approach is the factor-$2$-approximation
of Gonzalez~\cite{gonzalez1985clustering}, which we refer to as \gon.
This algorithm chooses an arbitrary vertex from the graph, and marks it as a center.
At each following step, the vertex farthest from the existing centers is marked as a new center,
until~$k$ centers have been chosen.
As the edge weights comprise a metric, the triangle inequality ensures that the resulting set of centers comprises a $2$-factor approximation.

\paragraph*{Parallelized version}
Given a point set~$V$ and a metric~$d$, with~$\OPT$ representing the optimal
covering radius,
Algorithm~\ref{alg:mr} obtains a set of centers~$\{c_i\}$
for which all points
in~$V_i$ -- where $\{V_i\}$ partitions~$V$ -- are within radius $2\cdot \OPT$ from~$c_i$. 
Running \gon on $C= \cup_i c_i$ obtains~$k$ centers whose
covering radius for~$C$ is~$2 \cdot \OPT$.
Assume that we have~$m$ machines each with capacity~$c$.
If $n/m \leq c$ and $k\cdot m\leq c$ then,
due to the triangle inequality, this results in a $4$-approximation MapReduce algorithm for $k$-center.
If the sample is too large to fit onto the final machine, further iterations of the first round can be run on the sample until there are few enough points.
Each additional round increases the approximation ratio by~$2$.

We dub this multi-round scheme
for $k$-center \ours, for ``MapReduce Gonzalez'', as shown in 
Algorithm~\ref{alg:mr}.

\begin{algorithm}
\begin{algorithmic}[1]
\small
\STATE $S \gets V$
\WHILE{ $|S| > c$}
\STATE The mapper arbitrarily partitions~$V$ into sets $V_1, \ldots, V_m$ such that $\cup_i V_i = V$ and $|V_i| \leq \lceil n/m \rceil $, and each set $V_i$ is sent to a reducer~$\rho_i$. 
\STATE In parallel, each reducer~$\rho_i$ runs \gon on~$V_i$, and returns a set~$C_i$ containing the $k$~centers found.
\STATE $S \gets \cup_i C_i$
\ENDWHILE
\STATE The mapper sends all points in $S$ to a single reducer.
\STATE This reducer runs \gon on $S$, and returns the set of centers $C^G$.
\RETURN $C^G$.
\end{algorithmic}
\caption{$\texttt{\ours}(V,k,m)$}
\label{alg:mr}
\end{algorithm}

\subsection{Approximation}

Algorithm \ours clearly
runs in polynomial time; to prove the four approximation of the $2$-round
case, we prove the following intermediate result.
For an arbitrary subset~$S$ of the vertex set~$V$,
let~$S^G$ denote the set of points in the solution 
obtained by running \gon on~$S$, and let $\mathit{SOL}_S$ denote the value,
the covering radius, of this solution.

\begin{lemma}\label{thm:approx}
For each $S \subseteq V$, $\mathit{SOL}_S\leq 2\cdot \OPT$.
\end{lemma}

\begin{proof}
Let~$V^*$ be an optimal set of centers.
The vertex set~$V$ can be partitioned into~$k$ sets $\{V^*_i\}_{i=1}^m$
such that  all points in set~$V^*_j$ are within $\mathit{OPT}$ of some
center $j\in V^*$.

First, assume every set $S^G\cap V^*_j$ has exactly one point.
This point,~$s_j$,
can serve as the center
for every point in~$V^*_j$.
Then every point~$x$ in $V^*_j$, and hence in $V_j^* \cap S$, is within
$2\cdot\mathit{OPT}$ of~$s_j$,
as both~$x$ and~$s_j$ are within $\mathit{OPT}$ of~$j \in V^*$.

However, if there is some partition~$V^*_j$
with $|S^G\cap V^*_j| > 1$, then points in the same partition are within $2\cdot \mathit{OPT}$ from each
other. \gon adds a new center to~$S^G$
only when it is the farthest from the points previously added to~$S^G$.
The presence of two centers within $2 \cdot \mathit{OPT}$ implies that all
points in~$S$ are within $2\cdot \mathit{OPT}$ of~$S^G$
(if there were some point farther, it would be in~$S^G$ instead).

Therefore, for every subset~$S$ of~$V$, the value of the $k$-center solution
returned by \gon on~$S$ is at most twice the optimal solution for~$V$.
\end{proof}

With sufficient space for~$C$,
the consequence of Lemma~\ref{thm:approx} is  a  factor-four approximation.
\begin{lemma}
\label{thm:mr}
If $n/ m \leq c$ and $k\cdot m\leq c$, then the $k$-center algorithm can be implemented in MapReduce in two rounds with a $4$-approximation guarantee. 
\end{lemma}

\begin{proof}
Let~$V_i$ refer to the points mapped to reducer~$\rho_i$.
Since we run \gon on~$V_i$, every point in~$V_i$ is within
$2\cdot\OPT$ of a center in~$C_i$ and hence in~$C$.

According to Lemma~\ref{thm:approx},
running \gon on~$C$ arrives at a set of centers~$C^G$
that is a $2\cdot\OPT$ solution on~$C$.
By the triangle inequality, it then follows that every vertex in the graph is
within $$2\cdot \OPT + 2\cdot \OPT = 4\cdot \OPT$$ of the $k$~centers~$C^*$.
\end{proof}

We now describe the properties of the setup and input for when \ours
can run effectively in two rounds.
The capacity required is $O(\max(n/m, k\cdot m))$,
based on which of the two rounds receives the most points. 
We assume that $n>k$, otherwise the solution to $k$-center is trivial.
We further assume that $n/ m > k$:
if this is not the case, then we can reduce
the number of machines.
For small~$k$, we only require that there is sufficient space across the machines to store the data set: 
that is, $m\cdot c \geq n$. 
We could also exploit external memory, for example by running multiple instances of our MapReduce algorithm and using a $k$-center
algorithm on the disjoint union of the solutions;
such cases are beyond the scope of this paper.

\subsection{Multi-round analysis}

If $k\cdot m> c$, we
lack the required memory to store the sample on a single machine, and therefore run further iterations of the while loop.  
In such instances, \ours uses more MapReduce rounds, loosening the approximation guarantee.

\begin{lemma}
\label{thm:mr_rounds}
If $n/ m \leq c$ and $k\cdot m\geq c$, then the $k$-center problem can be implemented in $i$ rounds with a $2(i+1)$-approximation, where~$i$ is chosen such that inequality~\ref{eq:inequality} is satisfied.
\end{lemma}

During each round, the number of centers is decreased, ending when they fit on a single machine:
each additional round adds two to the approximation factor.
As $k\cdot m> c$ and $m\geq 1$, it follows that $k>c$.
Even relaxing the requirement that $k\cdot m\leq c$, it is still necessary
that $k\leq c$.
Without this condition, selecting~$k$ centers
from a single machine seems to require incorporating external memory in some manner.

Assuming that $n/m\leq c$, after the first round we have $k\cdot m$ centers, so we send them to $m' = \lceil(k\cdot m)/c\rceil \leq  (k\cdot m)/c + 1$ machines.
After the second round, we have $k\cdot m'$ centers, which we can send to $m'' \leq  \lceil(k\cdot m')/c\rceil \leq m\cdot k^2/c^2 + k/c +1$ machines.  
In general,
the number of machines required after~$i$ rounds observes the bound
\begin{equation}\label{eq:inequality}
m^{(i)} \leq m\cdot (k/c)^i + \frac{1-(k/c)^i}{1-k/c}\,,
\end{equation}
and we can run the final round when $m^{(i)}<2$.
As~$i$ increases, the second term in the inequality approaches ${1}/[1-(k/c)]$, which
itself will be less than~$2$ only if $2k<c$.
Intuitively, during each round we select $k$ centers from each of the machines, so if~$k$ is close to~$c$ then the reduction in the number of centers in each round will be small.

\section{Revisiting the sampling approach}

In this section, we introduce a generalization of
Ene et al.'s~\cite{ene2011fast} iterative-sampling procedure.
As we show below, their algorithm is slower, but is more effective on the whole, than the sequential and parallel versions of Gonzalez's algorithm.
So that we can trade off runtime with approximation ratio,
we add a new parameter to the iterative-sampling approach, and call this
generalization \ene.
Before this, we make some alterations to the scheme to prevent eccentric behaviors that sometimes occur.

\subsection{Termination}
The core of Ene et al.'s scheme is shown as Algorithm~\ref{alg:EIM-MRS}.
Our implementation this algorithm adjusts the removal of points from~$R$
to ensure that the size of the set decreases in every iteration.
For our implementation of line~\ref{step:keepx},
we remove vertices whose distance from~$S$ is \emph{equal} to that
from~$v$ to~$S$. 
In the original presentation such a vertex would remain in~$R$,
which might lead to iterations in which
no vertices are removed from~$R$, and the procedure looping indefinitely. 

{
\begin{algorithm}
\begin{algorithmic}[1]
\small
\STATE $S\gets \emptyset, R\gets V$
\WHILE{$|R|>(4/\eps)kn^{\eps}\log n$}
		\STATE The mappers partition~$R$, and uniquely map each set $R^i$ to a reducer~$i$.
		\STATE Reducer~$i$ independently adds each point in~$R^i$ to set~$S^i$  with probability $9kn^{\eps}(\log n)/|R|$, and to set~$H^i$  with probability $4n^{\eps}(\log n)/|R|$.
		\STATE Let $H:=\bigcup_{1\leq i\leq\lceil n^{\eps}\rceil}H^i$ and $S:=S\cup(\bigcup_{1\leq i\leq\lceil n^{\eps}\rceil}S^i)$. The mappers
assign~$H$ and~$S$ to one machine, along with all edge distances between~$H$ and~$S$.
		\STATE The reducer sets $v\gets \texttt{Select}(H,S)$.
		\STATE The mappers arbitrarily partition~$R$, with~$R^i$ denoting  these sets. Let~$v$, $R^i$, $S$, and the distances between~$R^i$ and~$S$ be assigned to reducer~$i$. 
		\STATE For $x\in R^i$, remove $x$ from $R^i$ if $d(x,S)\leq d(v,S)$.\label{step:keepx}	 	
		\STATE Let $R:=\bigcup_{i}R^i$.
\ENDWHILE
\STATE Output $C:=S\cup	R$.
\end{algorithmic}
\caption{$\texttt{EIM-MapReduce-Sample}(V,E,k,\eps)$}
\label{alg:EIM-MRS}
\end{algorithm}
}

With relatively small graphs, there is
a non-trivial probability that the point~$v\in H$ will also be in~$S$.
In such cases, the vertex~$v$ will be at equal distance to~$S$ as the
points prior to it in the ordering given in line~\ref{step:whileR}
of $\texttt{Select}()$ (Algorithm~\ref{alg:select}).
This would mean that even points added to the sample might not be removed from~$R$.  
This increases the relative size of $R\, \cap\, S$, also increasing
the probability of no vertices being removed from~$R$ in subsequent rounds,
as~$H$ is sampled from~$R$.
If all points in~$R$ are eventually added to~$S$ then the algorithm cannot
terminate.
Therefore we assume that sampled points should \emph{always} be removed
from~$R$, and as such have adapted the algorithm to reflect this.

\begin{algorithm}[H]
\begin{algorithmic}[1]
\small
\STATE For each point $x\in H$, find $d(x,S)$
\STATE \label{step:whileR} Order the points in~$H$ according to their distance to~$S$ from farthest to smallest.
\STATE Let~$v$ be the point in position $\phi(\log n)^{th}$  in the ordering.
\RETURN $v$
\end{algorithmic}
\caption{$\texttt{Select}(H,S)$, with our parameter~$\phi$.}
\label{alg:select}
\end{algorithm}

\subsection{Trade-off}
Ene et al.~\cite{ene2011fast} prove that with high probability their MapReduce
procedure runs in $O(1/\eps)$ rounds.  
To decrease the number of rounds, we introduce parameter~$\phi$ to
$\texttt{Select()}$, which
trades off approximation for running time.
The original algorithm effectively fixed~$\phi$ to be~$8$.

In the original \ene scheme,
the expected number of points in~$R$ that are farther from~$S$ than~$v$ is
$8\log n\cdot|R|/|H|=|R| \cdot 2/n^{\eps}$.
By choosing a lower threshold for point~$v$, we decrease the number of points that remain in~$R$.
Since the sampling algorithm terminates when~$|R|$ falls below a the threshold defined by~$v$,
potentially this decreases the number of iterations.
We introduce a variable~$\phi$, and choose~$v$ such that it is
the $\phi(\log n)^{\text{th}}$ farthest point in~$h$ from~$S$.

To obtain a feasible $k$-center solution from the sample given by 
\texttt{EIM-MapReduce-Sample()},
a sequential $k$-center procedure is run on the resulting sample
in an additional MapReduce round. 
Note that in line~3,~$R$ is partitioned into $\lceil|R|/n^{\eps}\rceil$ sets of size
at most $\lceil n^{\eps}\rceil$, and in line~7 the mappers partition  $R$ into $\lceil n^{1-\eps}\rceil$ sets of size at most  $\lceil|R|/n^{1-\eps}\rceil$.
In Section~\ref{sec:sampling}, we prove that -- with weaker bounds, and with
appropriate values of~$\phi$ -- the probabilistic
$10$-approximation still holds.

\section{Runtime analysis}
\label{sec:runtime}

We now analyze in detail these parallel algorithms for $k$-center.
Ene et al.~\cite{ene2011fast} proved that their sampling procedure required $O(1/\eps)$ rounds with high probability, while \ours can run in two rounds given sufficient resources.
We consider also the computations required in each of the rounds to determine the expected overall runtime.

\subsection{\ours}
Assuming that $n/ m \leq c$ and $k\cdot m\leq c$, \ours will run in two consecutive MapReduce iterations.
The first iteration involves running~$m$ concurrent $k$-center algorithms,
each on $n/m$ vertices. 

The runtime of \gon on~$N$ points is $O(k\cdot N)$: each time a new center is selected, we need to find the distance of that center to all of the other vertices.
So the runtime for the first round of \ours is $O(k\cdot n/m)$, with a low
constant in the~$O(\cdot)$ expression.
In its second round, \ours runs \gon on the $k\cdot m$
centers obtained from the first round; this gives a runtime of $O(k^2\cdot m)$. 
Therefore the total runtime of \ours is $O(k\cdot n/m + k^2\cdot m)$, and 
for larger data sets, we would expect the dominant term to be $kn/m$.

\subsection{\ene}
The sampling algorithm, \ene, has, with high probability,~$T\in \Theta(1/\eps)$
iterations -- each comprising
three MapReduce rounds --
followed by a final clean-up round at the end that solves a single $k$-center
instance. 
Let~$R_{\ell}$  and~$S_{\ell}$ denote the state of sets~$R$ and~$S$,
respectively, in iteration~$\ell$ of the main loop of the algorithm.
Counting from the first iteration, $|R_0|=n$ and, with high probability, $|R_\ell| = O(n/n^{\ell\eps})$.
In each iteration, points in~$R$ are added to~$H$ with probability
$4n^{\eps}(\log n)/|R|$, so $|H|$ is expected to be~$O(n^\eps \log n)$.
And in line~5, $|S_\ell|$ becomes $|S_{\ell-1}| + O(kn^\eps\log n)$, so that,
starting with $|S_0|=0$, we expect 
$|S_\ell| = O((\ell+1) kn^\eps\log n)$.
We now analyze each MapReduce round.

\vspace{2mm}
\paragraph*{Round 1}
(Lines $3$ \& $4$). 
This round involves $O(|R_{\ell}|/m)$ operations during iteration~$\ell$,
so the total number of operations is
\vspace{-2mm}
$$
\sum_{\ell< T} \frac{|R_\ell|}{m} \in \frac 1 m O\left(\sum_{\ell< T} \frac n {n^{\ell\eps}}\right)
\in O\left(\frac 1 m \cdot \frac n {1-n^{-\eps}}\right)\,.
$$
\vspace{-6mm}

\paragraph*{Round 2}
(Lines $5$ \& $6$). 
This uses $O(|H|\cdot|S_\ell|)$ distance calculations per iteration,
which is $$O((\ell+1) k(n^{\eps}\log n )^2)\,,$$
for each of
the~$T \in O(1/\eps)$ (w.h.p.) iterations.
The total runtime is $O(k(n^{\eps}\log n)^2\sum_{\ell < T}\ell)$, which is in
$$O((k/\eps^2)(n^{\eps}\log n)^2)\,.$$
\vspace{-10mm}

\paragraph*{Round 3}
(Lines $7$ \& $8$). 
The third round requires $O(|R_{\ell}|\cdot|S_{\ell}|/m)$ distance calculations in each
iteration, so summing over all iterations this is
\vspace{-2mm}
$$
  O \left( \frac 1 m k n^{1+\eps}\log n\sum_{\ell < T} \frac{\ell+1}{n^{\ell\eps}}\right)
 \subseteq O\left(\frac{k}{m} \cdot \frac{n^{1+\eps}\log n}{(1-{n^{-\eps}})^2} \right)\,.
$$\vspace{-6mm}

\paragraph*{Final round}
 This sends~$|S_T|$ points to a single machine,
on which, say,  \gon is run.
With high probability,
this takes time $$O(({k}/{\eps})n^{\eps}\log n \cdot k)=
O(({k^2}/{\eps})n^{\eps}\log n)\,.$$

In practice, we find the dominant procedure for \ene is Round~3, in which points are
removed from~$R$.
This makes sense, as in most cases $k\cdot n^{1+\eps}$ is much larger than
$k^2 n^{\eps}$; the converse inequality implies that $k > n$.
Furthermore, \ours also has $O(k\cdot n/ m)$ complexity for cases where $k\cdot m<n/m$.

Experiments confirmed that
the dominant round for each algorithm was linear in~$k$, rather
than quadratic.
Comparing the dominant round of \ene to \ours,
we expect  \ene to be slower by a factor of
$n^{\eps}(1-{n^{-\eps}})^{-2} \log n$.

\section{Approximation ratio of \ene}
\label{sec:sampling}

In this section, we prove that the approximation bound for
\ene still holds under our changes to the sampling algorithm. 
This proof is based on the original analysis,
with the necessary details altered for correctness and relevance.
The main impact of the parameter~$\phi$ is to vary the number of points we consider to be represented by the existing sample. 
To analyze this,
we first require a formal description of what it means for points to be well represented (or \emph{satisfied}) by a sample.

Let~$Y$ denote some subset of the vertex set~$V$ (that we hope will
`represent'~$V$).
Each vertex $x\in V$ is assigned to its closest point in~$Y$,
breaking ties arbitrarily but consistently.
For each $x\in V$, let $A(x,Y)$ denote the point in~$Y$ to which~$x$ is
assigned; if~$x$ is in~$Y$ we assign~$x$ to itself. 
For $y\in Y$, let $B(y,Y)$ denote the set of all points assigned to~$y$
when~$Y$ is the `assignment subset'. 

For a sample~$S$,
Ene et al. say that a point~$x$ is \emph{satisfied} by~$S$ with respect
to~$Y$ if $$\min_{y\in S}d(y,A(x,Y))\leq d(x,A(x,Y))\,.$$
In words, there is a point in the sample that is closer to its $Y$-assigned
point, than~$x$ is to its $Y$-assigned point, $a(x,Y)$.

Let~$SOL_E$ denote the set of points returned by $\texttt{EIM{-}MapReduce{-}Sample}$. 
The set~$SOL_E$ might not include every center in~$Y$, but a point~$x$ can still be satisfied if $a(x,Y)\notin SOL_E$
by including a point in~$SOL_E$ that is closer to~$x$ than $A(x,Y)$. 
If $\texttt{EIM{-}MapReduce{-}Sample}$ returns a set 
that satisfies every point in~$V$,
then the sample is very representative of the initial data set,
and the clustering algorithms based on it should perform well.
However, there is no guarantee that the sample satisfies all points;
instead we can only be guarantee that the number of unsatisfied points
is small, and their contribution to the performance of the clustering algorithms is negligible compared to that of the satisfied points.

The sets at the core of $\texttt{EIM{-}MapReduce{-}Sample}$
change with every iteration. 
Denote the state of  sets~$R$,~$S$ and~$H$ at the beginning 
of iteration~$\ell$ by~$R_\ell$, $S_\ell$,
and~$H_\ell$ respectively,
where~$R_0=V$ and~$S_0=\varnothing$.
The set of points that are removed from~$R$ during iteration~$\ell$ is
denoted by~$D_\ell$, so $R_{\ell+1}=R_\ell\setminus D_\ell$. 
Let $U_\ell$ denote the set of points in $R_\ell$ that are not satisfied by $S_{\ell+1}$ with respect to $Y$. 
Let $U$ denote the set of all points that are not satisfied by $SOL_E$ (the sample returned by the algorithm) with respect to $Y$. 
If a point $x$ is satisfied by $S_\ell$ with respect to $Y$, then it is also satisfied by $SOL_E$ with respect to $Y$, and therefore $U\subseteq \bigcup_{\ell\geq 1}U_\ell$.

From the analysis of Ene et al.~\cite{ene2011fast} we have the following lemma.
\begin{lemma}
\label{unsatisfied}
Let~$Y$ be an arbitrary set with no more than~$k$ points. 
In iteration~$\ell$ of \sloppy\texttt{EIM{-}MapReduce}\texttt{{-}Sample}, where $\ell\geq0$,
$$\mathbf{P}\hspace{-1mm}\left[|U_{\ell}|\geq |R_{\ell}|/3n^{\eps}\right]\leq
n^{-2}\,.$$
\end{lemma}

We now show that our adaptation of the sampling algorithm by Ene et al. retains the same probabilistic $10$-approximation guarantee.

\begin{lemma}
\label{chernoff}
Let~$Y$ be a set of no more than~$k$ points.
In iteration~$\ell$ of the while loop in
\sloppy\texttt{EIM{-}MapReduce}\texttt{{-}Sample},
let~$v_\ell$ denote the threshold in the current iteration:
the  point in~$H_\ell$ that is the~$\phi\log(n)^{th}$ most distant
from~$S_{\ell+1}$.
Then there exist values~$a$ and~$b$ such that, for some~$\gamma > 0$,
$$\mathbf{P}\left[\frac{a|R_\ell|}{n^{\eps}}\leq |R_{\ell+1}|
\leq\frac{b|R_\ell|}{n^{\eps}}\right]\geq 1-\frac{2}{n^{1+\gamma}}\,.$$
\end{lemma}

\begin{proof}
\label{chernoff-proof}
Recall that we selected a pivot point~$v$, and discarded the points that are well represented by the current sample, compared to~$v$. 
Note that $R_{\ell+1}$ is the set of points in $R_\ell$ such that the distance to the sample $S_{\ell+1}$ is greater than the distance between the pivot point and the sample.

Ene et al. introduce these handy definitions.
For a vertex~$t$, we refer to the number of points in~$R$ further
from~$S_{\ell+1}$ than the point~$t$ as the~$\mathit{rank}_R$
of~$t$. 
For some value~$i$, and a set $Y\subseteq R$,
define $L(i,Y)=\big\vert{\{x\in Y : \mathit{rank}_R(x) \leq i\}}\big\vert$
as the number of points in the set~$Y$ that have rank smaller than~$i$.

\noindent
Let $r={d|R_\ell|}/{n^{\eps}}$, and let $|H_\ell| = c\cdot n^{\eps}\log n$.  
By design,
$$\mathbf{E}[ L(a\cdot r, H_\ell)] = {a\cdot
r\cdot|H_\ell|}/{|R_\ell|} = acd \log n$$
and $$\mathbf{E}[ L(b\cdot r, H_\ell)] = {b\cdot
r\cdot|H_\ell|}/{|R_\ell|} =  bcd \log n\,.$$
If $a\cdot r\leq |R_{\ell+1}| \leq b\cdot r$, then with high probability, the pivot
(chosen to be the  $\phi\log n^{\text{th}}$ point in $H_{\ell}$) will be in the range $[L(a\cdot r,H_{\ell}), L(b\cdot r,H_{\ell})]$.

\noindent
By the Chernoff inequality, 
\begin{align*}
 \mathbf{P}&[L(a\cdot r, H_\ell)\geq \phi\cdot\log n] \\
&= \mathbf{P}[L(a\cdot r, H_\ell)\geq(1+\delta)E[L(a\cdot r, H_\ell)]] \\
&= \mathbf{P}[L(a\cdot r, H_\ell)\geq(1+\delta)\cdot a c d\log n]  \\
&\leq \exp\left[\frac{-\delta^2\cdot a c d\log n}{2+\delta}\right] \\
&= n^{-{\delta^2\cdot acd}/(2+\delta)}.
\end{align*}
\vspace{-2mm}
 
\noindent
Choosing~$\delta$ so that
$(1+\delta)E[L(a\cdot r, H_\ell)] \leq
\phi\cdot\log n$
gives $\delta \leq -1+\phi/(a c d)$. 
Since the Chernoff bound requires that $\delta>0$,
we insist that $\phi>acd$.

The lemma statement requires $\mathbf{P}[L(a\cdot r, H_\ell)\geq \phi\cdot\log n ]\leq n^{-(1+\gamma)}$, which we can achieve by finding values
of $a$, $c$, $d$, and  $\phi$ that for some $\gamma > 0$ satisfy
$$\frac{({\phi}/(a c d)-1)^2\cdot acd}{(2+({\phi}/(a c d)-1))}\geq
(1+\gamma)\,.$$
Letting $x = 1+\gamma$, this is equivalent to $(acd)^2-(2\phi+x)acd+(\phi^2-x\phi)\geq 0$, which has real roots at $acd = \phi+x/2 \pm\sqrt{2x\phi+x^2/4}$.
Similarly, 
\begin{align*}
\mathbf{P}&[L(b\cdot r, H_\ell)\leq \phi\cdot\log n]\\[2pt]
  &= \mathbf{P}[L(b\cdot r, H_\ell)\leq(1-\delta)E[L(b\cdot r, H_\ell)]] \\[2pt]
    &= \mathbf{P}[L(b\cdot r, H_\ell)\leq(1-\delta)\cdot bcd\log n ]\\[2pt]
 & \leq \exp\left[\frac{-\delta^2\cdot bcd\log n }{2}\right] \\[2pt]
   & = n^{-{\delta^2 bcd}/{2}}  \leq n^{-x}.
\end{align*}
\vspace{-2mm}

Choosing~$\delta$ so that $(1-\delta)E[L(b\cdot r, H_\ell)] \leq
\phi\cdot\log n$ gives $\delta \leq 1-\phi/(b c d)$.
 Since the Chernoff bound requires that $\delta>0$, this gives the
constraint $\phi<bcd$.

For the last inequality to hold,
we need to find values of~$b$,~$c$,~$d$, and~$\phi$ such that  
$$(bcd)^2-(2\phi+2x)bcd+\phi^2\geq 0\,.$$
This has real roots at $bcd = \phi+x\pm\sqrt{2x\phi+x^2}$.

This gives feasible solutions for $acd \leq \phi+x/2 -\sqrt{2x\phi+x^2/4}$ and $bcd \geq \phi+x+\sqrt{2x\phi+x^2}$.
For later results we require that $a=1$ and $b\leq 5$. 
So for there to exist feasible values of~$c$ and~$d$, we have the following constraint,
\begin{align}\label{eqn:bound}
\frac{\phi+x+\sqrt{2x\phi+x^2}}{b} \leq    \phi+\frac x 2 -\sqrt{2x\phi+\frac{x^2}{4}},
\end{align}
where $b\leq 5$ and $x =  1+\gamma$.
When this bound holds, we can find values of each of the parameters such that the probability of $|r_{\ell+1}|$ being outside of the defined bounds is less than $1/n^x$, and 
therefore 
$$\mathbf{P}\left[\frac{|R_\ell|}{n^{\eps}}\leq
|R_{\ell+1}|\leq\frac{5|R_\ell|}{n^{\eps}}\right]\geq 1-1/n^{x}-1/n^{x}\,,
$$
which is $1-{2}n^{-(1+\gamma)}$.
With this probability,
the number of points in~$R_\ell$ that are further from~$S$ than~$v_\ell$
(and hence the size of the set~$R_{\ell+1}$)
is in the range $[{|R_\ell|}/{n^{\eps}},{5|R_\ell|}/{n^{\eps}}]$.
\end{proof}

Ene et al.~\cite{ene2011fast} prove that with probability $1-O(1/n)$, it is possible to map each unsatisfied point to a satisfied point such that no two unsatisfied points are mapped to the same satisfied point. 
Such a mapping allows them to bound the cost of the unsatisfied points with regards to the cost of the optimal solution. 
Their proof relies on the choice of $b=4$, and the bound from Lemma~\ref{chernoff} giving a probability greater than $1-2n^{-2}$. 
However, we use $b\leq 5$, and only assure a probability of $1-2n^{-(1+\gamma)}$. 
Therefore, we prove that the required mapping exists with probability $1-O(n^{-(1+\gamma)})$; by setting $\gamma = (\log\log n)/\log n$ this gives a probability of  $1-O(1/\log n)$, which is sufficient for large values of $n$. 
The choice of~$b$ arises from the requirement that $b/n^{\eps}<2$:
for $\eps =0.1$ and $n\geq 10,\!000$, this holds for $b\leq 5$.

In the original analysis, Ene et al. proved that their results hold \emph{with high probability}, which they define as having probability $\geq 1-O(1/n^2)$.
We instead bound our confidence in these results with probability $1-O(1/n^{1+\gamma})$ for a variable $\gamma$, which we will refer to as \emph{with sufficient probability}, or \emph{w.s.p.}.

The following results follow from the above analysis and that given by Ene et al~\cite{ene2011fast}. 

\begin{lemma}
For the sample $S$ returned by $\texttt{EIM{-}MapReduce{-}Sample}$,
\emph{w.s.p.} we have $\OPT(V,S) \leq 5\cdot \OPT$.
\end{lemma}

\begin{lemma}
The procedure resulting from running an $\alpha$-approximation algorithm on the sample returned by  $\texttt{EIM{-}MapReduce{-}Sample}$ achieves a $4\alpha+2$-approximation for the $k$-centre problem   \emph{w.s.p.}
\end{lemma}

When running a $2$-approximation algorithm on the sample, this result gives a $10$-approximation bound on the resulting procedure. 
To achieve a success probability higher than $1-O(1/n)$, we need $x\geq 1$: by
the bound in Inequality~\eqref{eqn:bound}, this implies that $\phi>5.15$.

\section{Experiments}

In this section we compare three algorithm (families), both in terms of speed
and effectiveness, and contrast the theoretical properties of these methods
(as shown in Table~\ref{table_compare}
with their empirical performance.

\begin{table}[!t]
\renewcommand{\arraystretch}{1.4}
\caption{Theoretical comparison of algorithms: Approximation factor
represented by~$\alpha$, run times are asymptotic, $O(\cdots)$.}
\label{table_compare}
\centering
\begin{tabular}{l|r|r|c}
\hline
Algorithm    & 
$\alpha$
& Rounds & Runtime \\
\hline
GON~\cite{gonzalez1985clustering} 
	& $2$  
		& n/a    
			& $k\cdot n$  \\
\ours 
	& $4$    
 		& 2 
			& $kn/m + k^2m$ \\[1mm]
\ene~\cite{ene2011fast}
	& $10$
		& $O(1 / \eps )$ 
			& \pbox{20cm}{\small $\displaystyle \frac{kn^{1+\eps}\log
				n}{m(1-{n^{-\eps}})^2}$} \\
\end{tabular}
\end{table}

\subsection{Setup}
We run experiments on three algorithm families, each of which we implement in the~C language.
First is the typically $2$-round algorithm, \ours;
second is (our Section-\ref{sec:sampling} generalization of) the sampling algorithm
of Ene et al.~\cite{ene2011fast}, \ene;
third is the (standard) sequential algorithm, \gon.
The latter, with its factor-$2$ approximation guarantee serves as an
effectiveness baseline.

For the sake of consistency with previous literature,
our method of implementing these algorithms mimics that of
Ene et al.~\cite{ene2011fast} in several ways.
In particular, we adopt a MapReduce approach, but do not record the cost of moving data between machines.
(As \ours involves fewer rounds, the expected cost of moving data between
machines would be less than for \ene.)
We simulate the parallel machines sequentially on a single machine,
taking the longest processing time of the simulated machines
as the processing time for that MapReduce round. 
For all parallel implementations, \gon is the subprocedure for selecting the final centers. 

The experimental system is a `commodity' machine,
with 8GB of main memory and an Intel\textregistered \\
Core\texttrademark i7-2600 CPU @ 3.40GHz.

\subsection{Experimental design}
Ene et al.~\cite{ene2011fast} generated synthetic data,
designed to have a fixed number of similarly sized clusters.
Moreover, they  tested their algorithm for values of~$k$ equal to the number of clusters.
We evaluate the algorithms over a range of values of~$k$ and vary the numbers of inherent clusters.
In practice, the number of clusters may not be known in advance, and the number of clusters required can be independent of the structure of the data.
To better determine how well these algorithms are likely to perform in practice,
we extend these experiments to test on  graphs with different underlying structures.

In all of the experiments, the distance is Euclidean, computed as required
from the locations of the points.
The $k$-center algorithm assumes a complete graph as input, and a matrix
representation of a graph, with all distances stored explicitly,
might result in a significant proportion of the data sent between machines being unnecessary.
The number of machines,~$m$, is fixed to~$50$, while~$n$ and~$k$ vary.
Our preliminary experimentation with the \ene algorithm, over a range of values of~$\eps$,
confirms that Ene et al.'s choice of~$\eps=0.1$ was good.

In Section~\ref{sec:sampling}, we introduced a parameter~$\phi$ to the \ene sampling
approach.
In our experiments, we test the effect of lowering~$\phi$ from its
``original'' value of~$8$, both in terms of runtime and effectiveness
Corresponding with our theoretical results in Section~\ref{sec:sampling} we
choose $\phi=6$; and to determine the robustness of the algorithm, we test
with $\phi=4$ and $\phi=1$, which are below the bound of $\phi=5.15$ that was given in Section~\ref{sec:sampling}.

\subsection{Data sets}

We test against a combination of real and synthetic data sets,
primarily in two and three dimensions,
but with several real data sets of larger dimension.
The data sets have a range of sizes, from $10,\!000$ through to
$1,\!000,\!000$ points, with varying degrees of inherent clustering.
Our synthetic  data sets have three different formats, viz.

\begin{description}
\setlength\itemsep{0.5em}
\item[\unif] The~$n$ points are uniformly distributed in a two-dimensional square.
\item[\gauss] 
The~$k'$ cluster centers -- where~$k'$ might not equal~$k$ -- are uniformly randomly generated in a unit cube.
The~$n$ points are distributed into these clusters uniformly at random, resulting in clusters of roughly similar size.
This helps determine the accuracy with which the procedures can
identify different clusters.
Distance from points to the cluster center follows a Gaussian distribution with
$\sigma=1/10$.
These data sets mimic those used in the experiments of Ene et al~\cite{ene2011fast}.
\item[\unbalanced]
An unbalanced arrangement, similar to~\gauss.
The distribution of points to inherent clusters is biased
such that around half of the points are in a single
(inherent) cluster; the distribution between the remaining clusters remains uniform.
\end{description}

We generate three graphs of each size and type, and run the algorithms twice over each data set, taking the average.
This gives a total of six results for each type of data set, over three different graphs.

We take real data sets from the UCI Machine Learning Repository~\cite{Lichman:2013}, over a wide range of sizes, applications and dimensions. 
We run four tests over each of the real data sets, and take the average result.
We include results for the $25,010$-point training set for the \poker data
set, and the~$10\%$ sample from the $4,000,000$-point \kdd 1999 data set.

\begin{figure}[!t]
\center
{\includegraphics[width=3.1in]{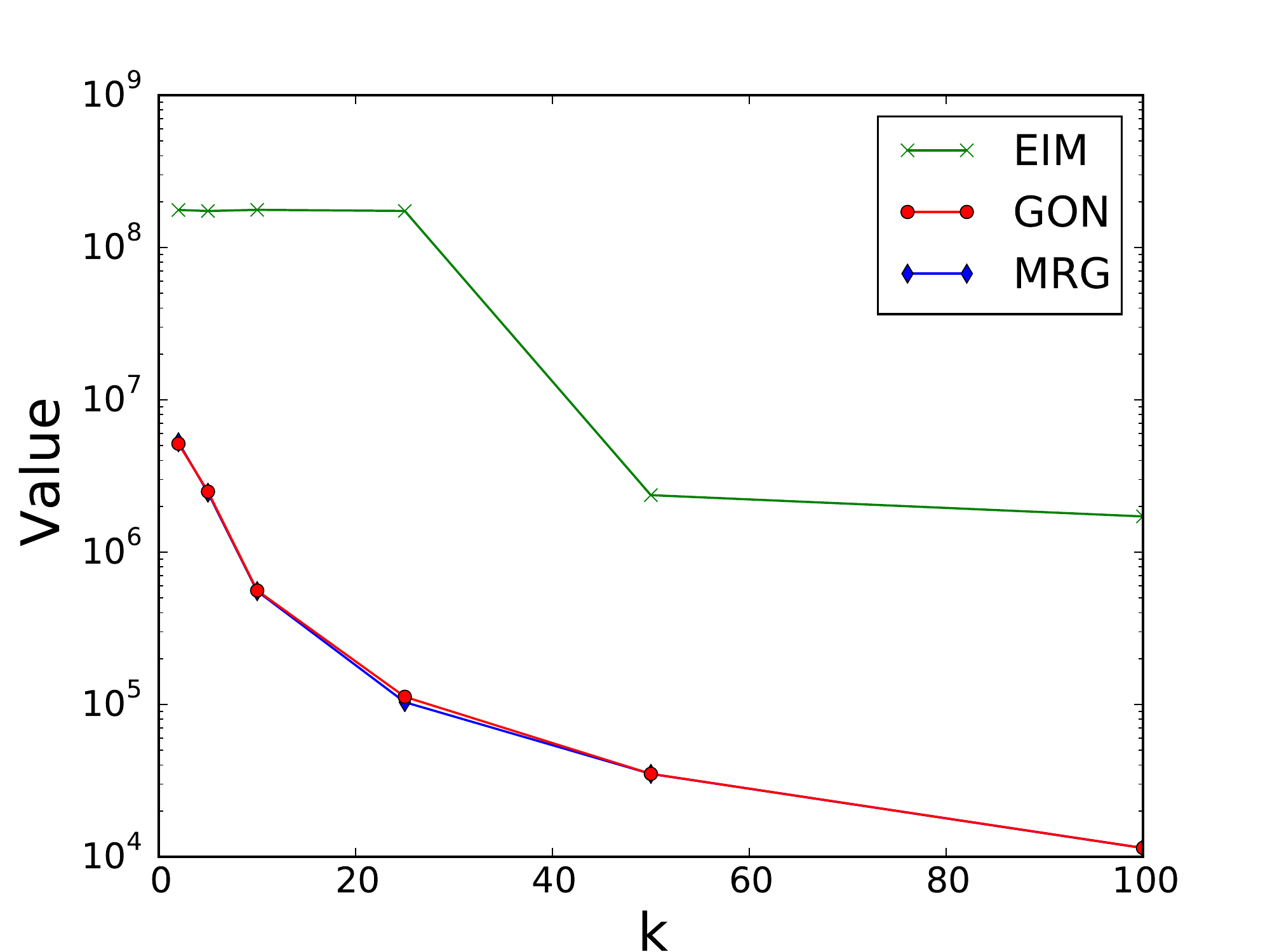}}
\caption{\small Solution values over~$k$
on \kdd 1999.}\vspace{-3mm}
\label{fig_real}
\end{figure}

\vspace{-2mm}
\section{Results}
\vspace{-2mm}
Overall \ours is faster than the alternative procedures - often by orders of magnitude,
with \ene running slower than the sequential algorithm despite being parallelized, conforming with the analysis in Section~\ref{sec:runtime}. 
\vspace{-1mm}

\begin{figure*}[!t]
\vspace{-4mm}
\centerline 
\hfil
{\subfloat[\gauss($n=1,\!000,\!000$, $k'=25$).]{\includegraphics[width=3.1in]{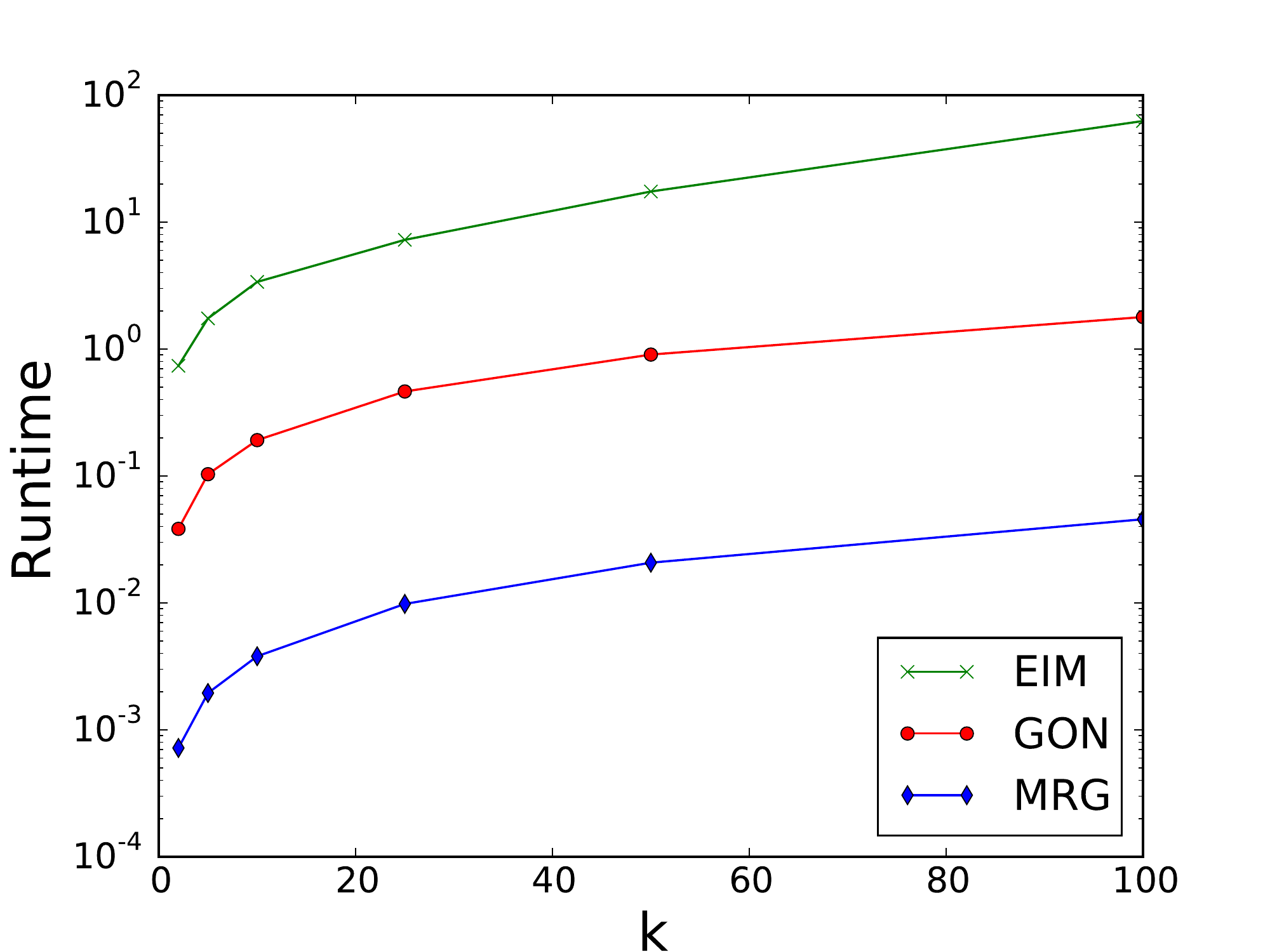}
\label{fig_time_first}}}
\hfil
\subfloat[\unif($n=100,\!000$).]{\includegraphics[width=3.1in]{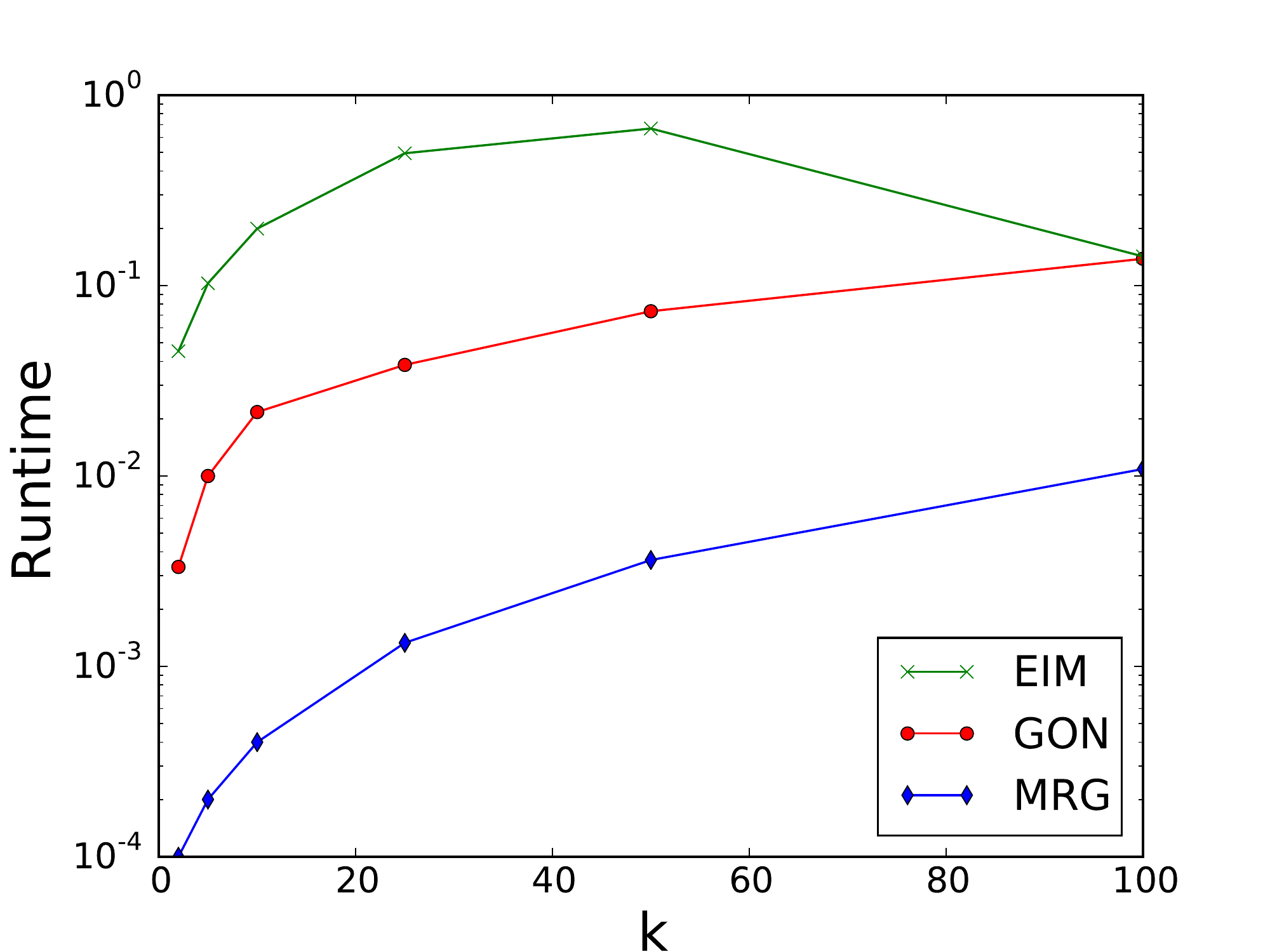}
\label{fig_time_second}}
\caption{\small Runtimes in seconds over a range of values
of~$k$.
Corresponding with our theoretical analysis, \ene runs slower than
both \ours and the sequential alternative, with \ours being the fastest of the algorithms considered.}
\label{fig_time}
\vspace{-4mm}
\end{figure*}

\begin{table}[!t]
\renewcommand{\arraystretch}{1.2}
\caption{\small Solution value over~$k$ for \gauss($n=1,\!000,\!000$, $k'=25$).}
\label{fig_value_first}
\centering
\small{
\begin{tabular}{r|S|S|S}
$k$&\ours\ \ &\ene\ \ &\gon \\ \hline
2  &96.04  &93.11  &95.86  \\
5  &61.90  &61.58  &63.31  \\
10 &41.31  &39.43  &39.72  \\
25  &0.961  &0.854  &0.961 \\
50  &0.762   &0.683  &0.719 \\
100  &0.607   &0.556  &0.573 \\
\end{tabular}
}
\renewcommand{\arraystretch}{1.2}
\caption{\small Solution value over~$k$ for \unif($n=100,\!000$).}
\label{fig_value_second}
\centering
\small{
\begin{tabular}{r|S|S|S}
$k$&\ours&\ene&\gon\\ \hline
2  &91.33  &95.80 &91.18  \\
5  &50.68  &50.65  &53.14  \\
10  &33.35  &31.12  &32.35  \\
25 &18.49  &18.01  &18.27  \\
50 &13.14 &12.39  &12.36  \\
100  &9.144  &8.764  &8.727 \\
\end{tabular}
}
\end{table}
\normalsize

\begin{figure*}[!t]
\centerline 
\hfil
{\subfloat[\gauss($n=1,\!000,\!000$, $k'=50$).]{\includegraphics[width=3.1in]{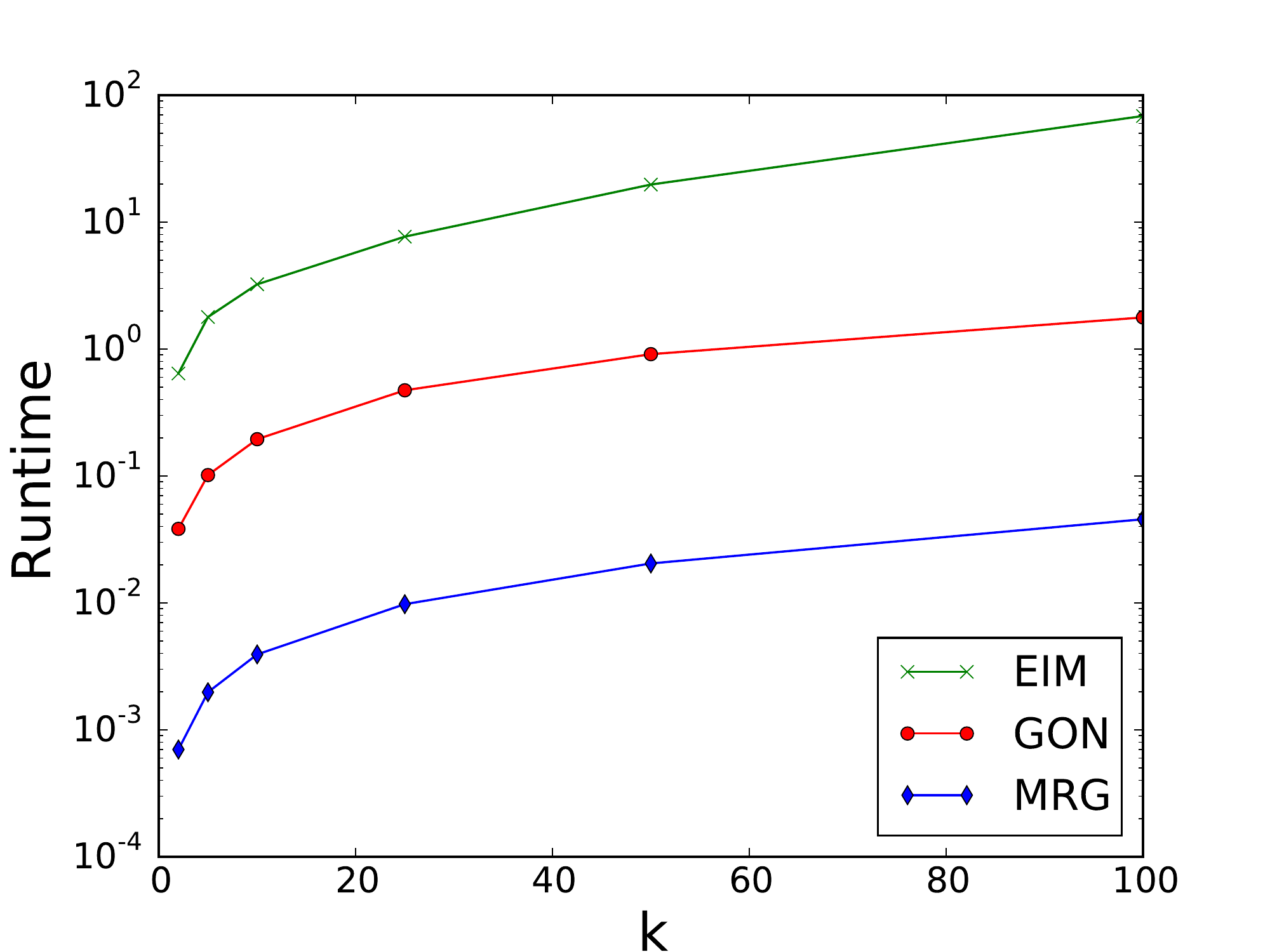}
\label{fig_size_first}}}
\hfil
\subfloat[\gauss($n=50,\!000$, $k'=50$).]{\includegraphics[width=3.1in]{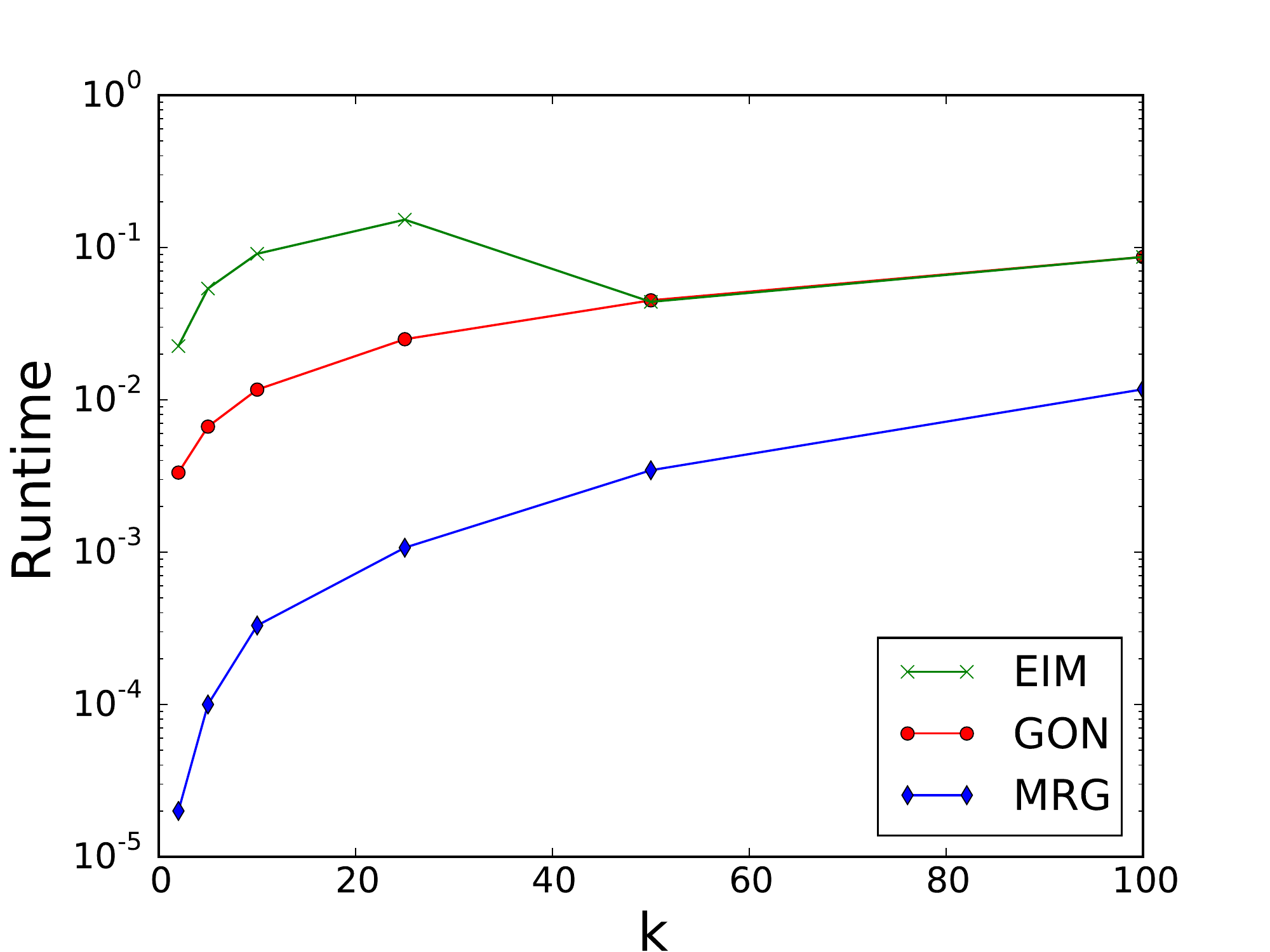}
\label{fig_size_second}}
\caption{\small Runtimes in seconds for \gauss graphs over a range values of~$k$.
When~$k$ becomes too large, relative to~$n$, \ene no longer performs sampling and defaults to the sequential algorithm.
}
\label{fig_size}
\end{figure*}

\begin{table}
\renewcommand{\arraystretch}{1.2}
\vspace{-4mm}
\caption{\small Solution value over~$k$ for \unbalanced($n=200,\!000$,
$k'=25$).
When~$k=k'$,
\ene is notably better.}
\label{compare_clust}
\centering
\small{
\begin{tabular}{r|S|S|S} 
$k$&\ours&\ene&\gon\\ \hline
2  &97.96  &93.69 &93.37  \\
5 &64.61  &64.28  &61.72 \\
10 &40.17 &40.05 &40.39  \\
25  &0.932  &0.828  &0.939 \\
50 &0.668 &0.643  &0.655 \\
100  &0.515  &0.530  &0.500 \\
\end{tabular}
}
\end{table}
\normalsize

\begin{table}[!t]
\renewcommand{\arraystretch}{1.2}
\vspace{-3mm}
\caption{\small Solution value over~$k$ for the \poker data set.}
\label{poker}
\centering
\small{
\begin{tabular}{r|S|S|S} 
$k$&\ours&\ene&\gon\\ \hline
2 &19.41  &18.60  &18.17 \\
5  &18.06   &17.07   &17.25   \\
10  &15.12  &14.20   &15.03   \\
25 &12.13   &11.98   &11.84   \\
50  &10.07  &9.418 &9.617 \\
100  &8.774 &9.241 &8.396 \\
\end{tabular}
}
\end{table}
\normalsize

\vspace{-2mm}
\subsection{Summary}
\vspace{-2mm}
In most cases, despite having worse approximation guarantees, the
solutions for the parallelized algorithms are comparable to
those of the baseline, \gon, with \ene performing slightly better for synthetic data sets.
Ene et al~\cite{ene2011fast} suggested that their sampling-based
algorithm did not perform particularly well, likely
due to the $k$-center procedure being sensitive to outliers. 
Our experimental results show otherwise: sampling fewer points can
occasionally provide \emph{better} results due to the tendency to
avoid sampling points that are well represented,
but toward the edge of the cluster. 
The tendency for \gon to favor outliers is often
 mitigated, rather than amplified, by sampling.
As shown in Table~\ref{compare_clust},
this effect is particularly evident for \gauss graphs where~$k=k'$.

As illustrated in Tables~\ref{fig_value_first} and~\ref{compare_clust},
for the synthetic data sets, the parallel algorithms
are about as effective as Gonzalez's algorithm.
In general, \ene is slightly more effective than \ours.
With the exception of the \ene results on the \kdd 1999~$10\%$ sample, for which it performs poorly,
the same occurs on the real data sets, as seen in Figure~\ref{fig_real} and Table~\ref{poker}.

\subsection{Running time}
For the majority of the experiments, \ene ran using two iterations of the main loop, for a total of seven MapReduce rounds.
On certain data sets,
\ene sometimes executes one iteration, sometimes two --
that is, four or seven
MapReduce rounds
 -- as the number of points removed per round is probabilistic.

From Figures~\ref{fig_time_second},~\ref{fig_size_second}
and~\ref{fig_second_case}, we can see that as the ratio of~$n$ to~$k$
drops, at some point, \ene merely sends the entire data set to a single
machine, rather than employing the sampling procedure.
We can also note that in Figure~\ref{fig_second_case}, \ours displays a different trend from Figure~\ref{fig_first_case}.
In Section~\ref{sec:runtime}, we showed that the runtime is $O(kn/m + k^2 \cdot m)$. 
For larger values of~$k$ and small values of~$n$, the $k^2\cdot m$ term
dominates;
as~$n$ grows, the $k\cdot n/m$ term dominates,
so the trend becomes similar to that in Figure~\ref{fig_first_case}. 
From our analysis in Section~\ref{sec:runtime}, both \ours and \ene have a
round with a $k^2$ term in the running time.
When~$k$ is large relative to~$n$,
this can potentially dominate.

\subsection{Runtime/Approximation Trade-off }

We examine the sensitivity of the \ene algorithm to the~$\phi$ parameter.
As expected, the variability of effectiveness increases
as the~$\phi$ parameter decreases,
while the runtimes significantly decrease. 
Tables~\ref{table_emi_sol} and~\ref{table_emi_run} compare
the average solution value and runtimes for the different parameters. 
The algorithm speeds up significantly 
for values of~$\phi$ below the threshold of~$5.15$ (above which there is a
guaranteed low probability of poor solutions, see Section~\ref{sec:sampling}).
Yet, in practice, it still returns acceptable solutions:
in some cases solutions are even better with smaller values of~$\phi$.

\begin{table}[!t]
\renewcommand{\arraystretch}{1.2}
\caption{%
Average solution value over~$\phi$, in \ene,
for \gauss  ($n=200,\!000$, $k'=25$). For each~$k$, the lowest value is in \emph{italics}.
}
\label{table_emi_sol}
\centering
{
\small
\begin{tabular}{r|S|S|S|S} 
& \multicolumn{3}{c}{$\phi$} \\
\emph{\textbf{$k$}}&\emph{\textbf{ $1$}}  &\emph{\textbf{ $4$}}  &\emph{\textbf{ $6$}}  &\emph{\textbf{$8$}}  \\ \hline
2 & 88.4  & $\mathit{80.4}$\hspace{1.5mm}  &85.5  &86.5  \\
5 & 59.9  &60.9  & $\mathit{56.5}$\  &61.9  \\
10 & 36.2  &35.5 & $\mathit{34.7}$\  &35.3  \\
25 & 0.796   &\hspace{2.5mm} $\mathit{0.780}$  & 0.826 &0.840 \\
50 & 0.630   &0.617   &\hspace{2.5mm} $\mathit{0.610}$  &0.666 \\
100 &\hspace{2.4mm} $\mathit{0.478}$  &0.492   &0.505  &0.535\\
\end{tabular}
}
\vspace{2mm}
\renewcommand{\arraystretch}{1.2}
\caption{%
Average runtime over~$\phi$, in \ene,
for \gauss ($n=200,\!000$, $k'=25$).
The lowest runtime in each row is in \emph{italics}.
}
\label{table_emi_run}
\centering
{
\small
\begin{tabular}{r|S|S|S|S} 
& \multicolumn{3}{c}{$\phi$} \\
\emph{\textbf{$k$}}&\emph{\textbf{$1$}}  &\emph{\textbf{$4$}} 
&\emph{\textbf{$6$}} &\emph{\textbf{$8$}}  \\ \hline
2& $\mathit{0.050}$ &0.059 &0.165 &0.135 \\
5& $\mathit{0.080}$ &0.130 &0.368 &0.314 \\
10& $\mathit{0.283}$ &0.480 &0.549 &0.552 \\
25& \hspace{-.4mm}0.588 &\hspace{3.4mm}$\mathit{0.505}$ &1.47   &1.42   \\
50& $\mathit{0.693}$ &0.816 &2.84  &2.24   \\
100&$\mathit{0.726}$ &0.757 &3.78   &3.59    \\
\end{tabular}
}
\end{table}

This seemingly counterintuitive behavior can be explained by the choice
of \gon as the sub-procedure for the sample. 
As noted above,
selecting the farthest points as new centers makes it more likely that points
at the perimeter of a cluster are chosen; although each cluster is well
represented by some vertex, the selected center is at the perimeter of the cluster. 
By sampling \emph{fewer} points, it is less likely that points that are
extremal to the cluster are present in the subgraph on which
\gon is run.
Therefore in decreasing the runtime of the algorithm,
for appropriate values of~$k$,
we can potentially improve the average value of the solutions obtained.
However this behavior is likely to be more volatile:
the guaranteed bound on the performance has lower probability,
giving a higher chance that a very poor solution is returned.

\section{Conclusion}
In this paper, we describe a multi-round parallel procedure for the~$k$-center
problem.
When it runs in only~$2$ MapReduce rounds, it is  a $4$-approximation. 
We show experimentally that it returns solutions that compare well
to those of a sequential $2$-approximation algorithm, while running extremely
fast. 

We compare this approach to the existing $10$-approximation
sampling-based MapReduce procedure~\cite{ene2011fast}.
It is often slightly more effective, but can be very slow.
To support our experimental results, we give the first detailed runtime analysis for the sampling approach, the proof of which correspond with our empirical results.
We also parameterized the
sampling procedure to improve runtimes, sometimes even bringing better
solutions despite the lack of a provable effectiveness bound.

\begin{figure*}[!t]
\vspace{-7mm}
\centerline 
\hfil
{\subfloat[$k=10$]{\includegraphics[width=3.1in]{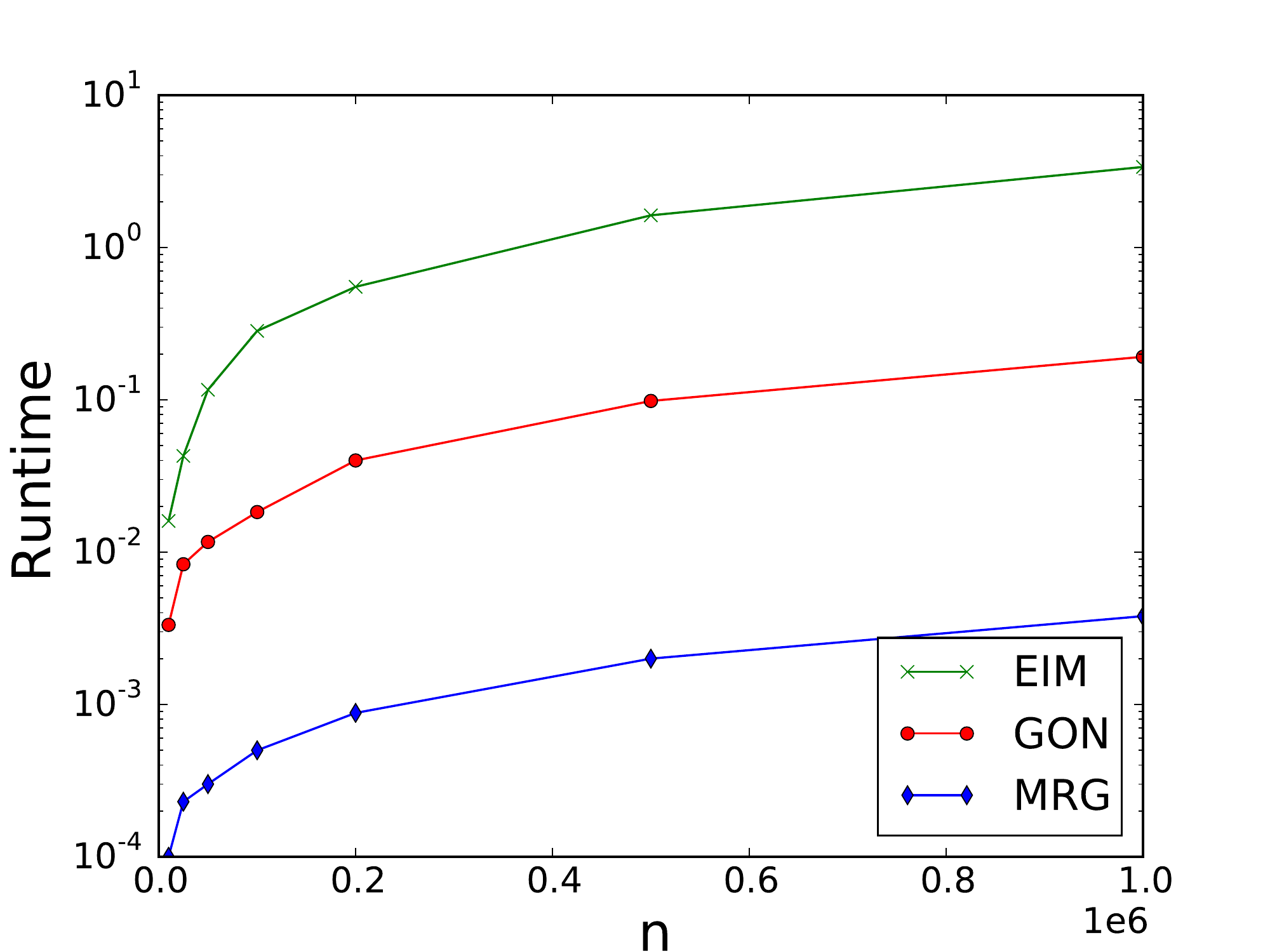}
\label{fig_first_case}}}
\hfil
\subfloat[$k=100$]{\includegraphics[width=3.1in]{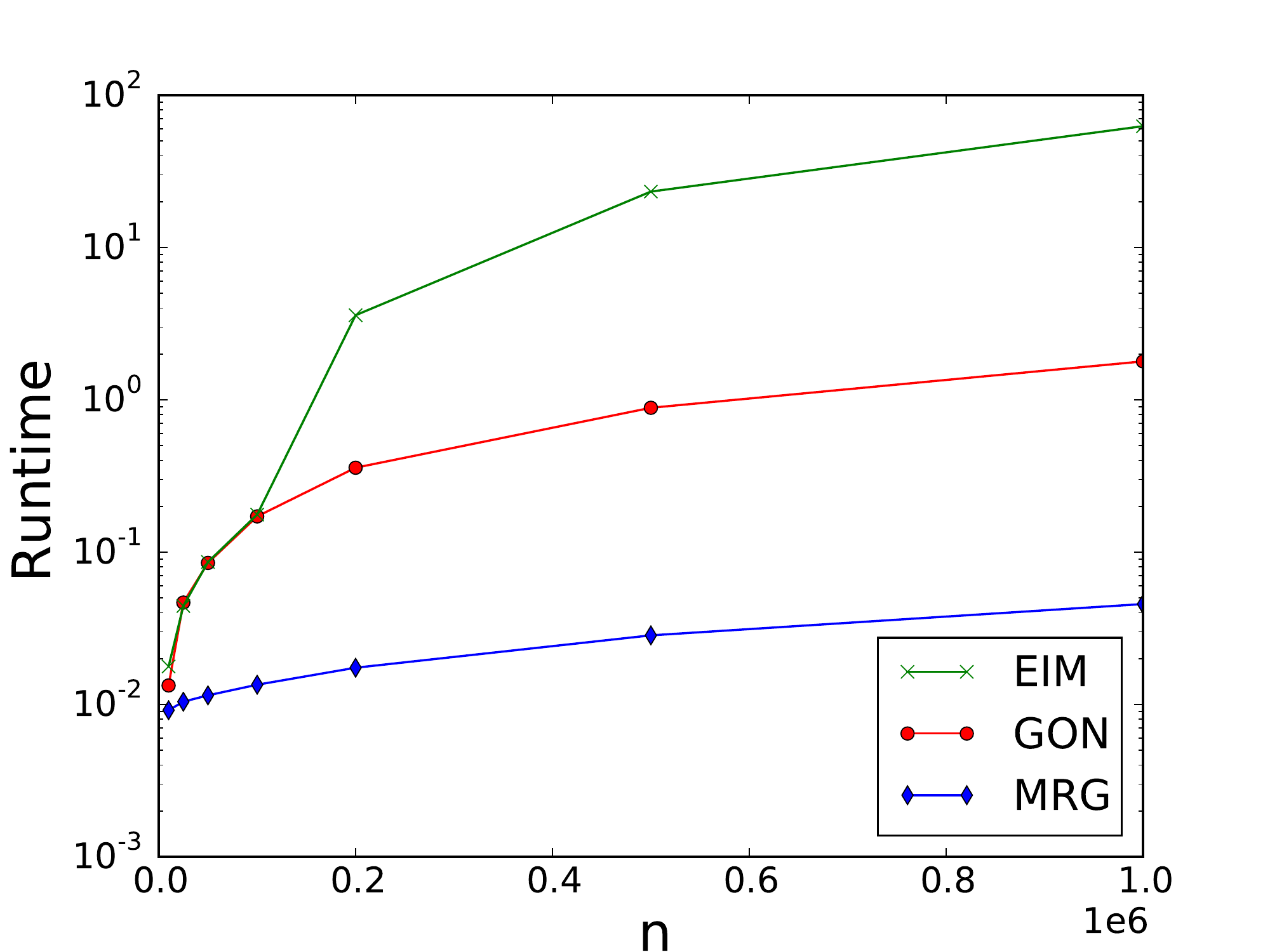}
\label{fig_second_case}}
\caption{\small Runtimes in seconds for fixed~$k$ over values of~$n$ ranging
from $10,\!000$ to $1,\!000,\!000$.
For sufficiently small values of~$n$ relative to~$k$, \ene behaves identically to \gon.
The is caused by the condition on the while loop: if~$k$ is large enough, the condition is never met and no sampling occurs,
so \gon is run on the entire data set.%
}
\label{fig_sim}
\end{figure*}

\subsection*{Future work}
The approximation factor of four for \ours is tight.
There are graphs on which,
with adversarial  assignment of points to machines and choice of seedings
for \gon,
 \ours gives a $4$-approximation.
How likely such cases are in practice?
We seek bounds on the probability that this algorithm gives a poor approximation.
And what is the effectiveness when \ours needs more than two rounds?

Recently, Im and Moseley~\cite{im2015} described a
$3$-round $2$-approximation MapReduce procedure for the $k$-center problem
under the assumption that $\mathit{OPT}$ is known, and announced a $4$-round procedure that does not require prior knowledge of the optimal solution -- 
these details have yet to appear.
More recently Malkomes et al.~\cite{malkomes2015fast} presented a parallel adaptation of the $k$-center algorithm comparable to a special case of our approach.
Currently all such approaches rely on the sequential algorithm of Gonzalez~\cite{gonzalez1985clustering}. 
It would be interesting to compare with similar adaptations of alternative sequential algorithms, such as that of Hochbaum \& Shmoys~\cite{hochbaum1985best}.

\vspace{3mm}
\paragraph*{Acknowledgment}
Many thanks to Kewen Liao for valued feedback and proofreading.

\newpage

\bibliographystyle{siam}
\bibliography{references}

\end{document}